\documentclass{amsart}
\usepackage{amsfonts}
\usepackage{blkarray}
\usepackage{float}
\usepackage{multirow}
\usepackage{amsmath}

\newtheorem{theorem}{Theorem}[section]
\newtheorem{lemma}[theorem]{Lemma}

\newtheorem{corollary}[theorem]{Corollary}
\theoremstyle{definition}
\newtheorem{definition}[theorem]{Definition}

\theoremstyle{remark}
\newtheorem{remark}[theorem]{Remark}
\numberwithin{equation}{section}
\newcommand{\F}{\mathbb{F}}
\newcommand{\R}{\mathcal{R}}
\newcommand{\Z}{\mathbb{Z}}

\begin{document}
\title[Codes over $\R_{k,m}$ and Applications to Binary self-dual codes]{On codes over $\R_{k,m}$ and constructions for new binary self-dual codes}
\author{Nesibe Tufekci}
\author{Bahattin Yildiz*}\thanks{*The corresponding author}
\address{Department of Mathematics, Fatih University, 34500, Istanbul, Turkey%
}
\email{ nesibe.tufekci@fatih.edu.tr, byildiz@fatih@edu.tr}

\begin{abstract}
In this work, we study codes over the ring $\R_{k,m}=\mathbb{F}_{2}[u,v]/$ $\left\langle u^{k},v^{m},uv-vu\right\rangle$, which is a family of Frobenius, characteristic 2 extensions of the binary field. We introduce a distance and duality preserving Gray map from $\R_{k,m}$ to $\F_2^{km}$ together with a Lee weight. After proving the MacWilliams identities for codes over $\R_{k,m}$ for all the relevant weight enumerators, we construct many binary self-dual codes as the Gray images of self-dual codes over $\R_{k,m}$. In addition to many extremal binary self-dual codes obtained in this way, including a new construction for the extended binary Golay code,  we find 175 new Type I binary self-dual codes of parameters [72,36,12] and 105 new Type II binary self-dual codes of parameter [72,36,12].
\end{abstract}

\maketitle

\section{Introduction}
Self-dual codes are an interesting family of codes that have attracted a lot of attention for their connections to many fields of study such as lattices, designs and invariant theory. Many construction methods have been pursued by researchers to find extremal and optimal binary self-dual codes. The works \cite{betsumiya}, \cite{bouyukliev}, \cite{bouyuklieva}, \cite{rdontcheva}, \cite{georgiou}, \cite{gulliverharada}, \cite{doublecirculant}, \cite{melchorgaborit} highlight some of these constructions. The previously mentioned constructions are generally done over finite fields. A recent serious of papers have shown that some suitable rings can be used to obtain binary self-dual codes.

Recall that in \cite{jwood}, the largest family of rings to study for coding theory were found to be the family of Frobenius rings. It was shown that the MacWilliams identities hold for codes over these rings. Codes over rings have been a focal point of research in the last two decades. Many different Frobenius rings were studied within that context for different reasons and motivations, leading to many different results. Among the oft-studied rings we can name $\Z_4$, $\Z_{p^k}$, Galois rings, finite chain rings, $\F_2+v\F_2$, $\F_2+u\F_2+v\F_2+uv\F_2$, $R_k$, etc.

The family of rings named $R_k$, which started with \cite{yildizkaradeniz}, have recently been used quite successfully to construct many good binary self-dual codes. This family of rings have provided an alternate method, to many existing ones, of constructing binary self-dual codes of different automorphism groups, and in many cases codes with new weight enumerators. (see \cite{karadeniz66}, \cite{karadeniz68}, \cite{karadenizyildizaydin}, \cite{BYSK} for example). The common theme in these works is the presence of a duality and distance preserving Gray map and the intricate structure of the ring with a high number of units that lead to large automorphism groups.

In this work, we introduce a generalization of rings of the form $\F_2+u\F_2+ \dots+u^k\F_2$ and $\F_2+u\F_2+v\F_2+uv\F_2$ to a family of rings that we denote by $\R_{k,m}$, where $\R_{k,m}=\mathbb{F}_{2}[u,v]/\left\langle u^{k},v^{m},uv-vu\right\rangle .$ Note that $\R_{1,1} = \F_2$, the binary field; $\R_{2,1} = \F_2+u\F_2$; $\R_{2,2} = \F_2+u\F_2+v\F_2+uv\F_2$ and $\R_{k,1} = \F_2+u\F_2+ \dots+u^{k-1}\F_2$. We establish that this is a Frobenius, characteristic 2, family of rings that is non-chain when $k$ and $m$ are both greater than $1$. We find a duality-preserving Gray map from $\R_{k,m}$ to $\F_2^{km}$, and using some of the common construction methods (the double circulant, bordered double circulant and four circulant constructions) of self-dual codes we find many good binary self-dual codes as the Gray images of self-dual codes over $\R_{k,m}$ for suitable $k$ and $m$. More precisely, we give an alternate construction to the extended Golay code; we find 6 of the 41 extremal binary self-dual codes of length $36$; 2 extremal self-dual binary codes of length $66$; 175 {\bf new} Type I binary self dual codes of parameters $[72,36,12]$ and 105 {\bf new} Type II binary self-dual codes of parameters $[72,36,12]$.

The rest of the work is organized as follows. Section $2$ includes the preliminaries about the structure of the ring. In Section 3, we introduce the Lee weight and the related distance-preserving Gray map, which we prove to be duality-preserving as well. In Section 4, we prove that MacWilliams identities for all the relevant weight enumerators. Section 5 contains our construction methods as well as the computational results, which are tabulated in the end.

\section{Preliminaries}

\subsection{The structure of the ring $\R_{k,m}$}

The ring $\R_{k,m}$ is defined as following for $k\geq m\geq 1$
\begin{equation*}
\begin{tabular}{l}
$\R_{k,m}=\mathbb{F}_{2}[u,v]/\left\langle u^{k},v^{m},uv-vu\right\rangle .$%
\end{tabular}%
\end{equation*}

$\R_{k,m}$ is characteristic $2$ ring of size $2^{km}$. When $k=m=1$ the ring
is $\mathbb{F}_{2}$. When $k=2,m=1$ the ring is $\mathbb{F}_{2}+u\mathbb{F}%
_{2}$ and Type II codes over this ring were studied in \cite%
{doughertygaboritharada}, when $k=m=2$ the ring is $\mathbb{F}_{2}+u\mathbb{F%
}_{2}+v\mathbb{F}_{2}+uv\mathbb{F}_{2}$ and codes over this ring were
studied in \cite{yildizkaradeniz}.

$\R_{k,m}$ can be viewed as an $\mathbb{F}_{2}-$vector space with a basis
\begin{equation*}
\begin{tabular}{l}
$\left\{ u^{i}v^{j}\mid 0\leq i\leq k-1,0\leq j\leq m-1\right\} .$%
\end{tabular}%
\end{equation*}%
Any element of $\R_{k,m}$ can be represented as
\begin{equation}\label{represent}
\begin{tabular}{l}
$\sum\limits_{\substack{ 0\leq i\leq k-1  \\ 0\leq j\leq m-1}}%
c_{ij}u^{i}v^{j},$ $c_{ij}\in \mathbb{F}_{2}$%
\end{tabular}%
\end{equation}%
in a unique way where, addition can be done in a natural way coordinate-wise
addition and multiplication of any two elements can be defined as%
\begin{equation*}
\begin{tabular}{l}
$xy{=}\sum\limits_{_{\substack{ 0\leq r\leq k-1  \\ 0\leq s\leq m-1}}%
}\left( \sum\limits_{_{\substack{ i_{1}+i_{2}=r  \\ j_{1}+j_{2}=s}}%
}c_{i_{1}j_{1}}d_{i_{2}j_{2}}\right) {u}^{r}{v}^{s}.$%
\end{tabular}%
\end{equation*}%
for any $x=\sum\limits_{\substack{ 0\leq i_{1}\leq k-1  \\ 0\leq j_{1}\leq
m-1}}c_{i_{1}j_{1}}u^{i_{1}}v^{j_{1}}$ and $y=\sum\limits_{\substack{ 0\leq
i_{2}\leq k-1  \\ 0\leq j_{2}\leq m-1}}d_{i_{2}j_{2}}u^{i_{2}}v^{j_{2}}\in
\R_{k,m}.$
Note that the sum of indices in the inner sum above is done with respect to modulus $k$ or $m$ where suitable. $\R_{k,m}$ is a finite commutative ring of characteristic $2$, of size $2^{km}$.
One of the important structural properties is to characterize the units and non-units in $\R_{k,m}$. The following lemma takes care of this:

\begin{lemma}\label{units}
An element in $\R_{k,m}$ of the form given in (\ref{represent}) is a unit if and only if $c_{00}$ is $1$.
\end{lemma}

\begin{proof}
Since the characteristic of the ring is $2$ and $c^{2^n} = c$ for all $c\in \F_2$ and $n\in \Z_+$, we have
\begin{equation*}
\begin{tabular}{ll}
$\left( \sum\limits_{0\leq i+j\leq k+m-2}c_{ij}u^{i}v^{j}\right) ^{2^{n}}$
& ${=}\sum\limits_{0\leq i+j\leq k+m-2}{c}_{ij}\left( u^{i}v^{j}\right)
^{2^{n}}$.
\end{tabular}%
\end{equation*}%
If we choose $n$ so that $2^n\geq k,m$, then the above sum becomes $c_{00}$. Thus, if $c_{00}=1$, this will make the element a unit, while when $c_{00}=0$, it will be a zero divisor and hence a non-unit.
\end{proof}

\begin{lemma}
The ring $\R_{k,m}$ is a local ring with unique maximal ideal $%
I_{u,v}=\left\langle u,v\right\rangle $. This ideal consists of all
non-units and has $\left\vert I_{u,v}\right\vert =\frac{\left\vert
\R_{k,m}\right\vert }{2}.$
\end{lemma}

\begin{proof}
Clearly, all non-units are in $I_{u,v}$ from Lemma \ref{units}. Since those are the
elements with $c_{00}=0$, we have that the cardinality of the ideals is half
the cardinality of ring.
\end{proof}

The maximal ideal $I_{u,v}$ is not generated by a single element, so the
ring $\R_{k,m}$ is not a principal ideal ring for $m>1$. Moreover, the ring is
a finite chain ring if $m=1$. Let us consider ideals $%
I_{u}=\left\langle u\right\rangle $ and $I_{v}=\left\langle v\right\rangle $
which are contained in $I_{u,v}$ but they are not related via inclusion. That is,
the ring is not a chain ring for $m>1$. One can also observe that $%
\R_{k,m}/Rad(\R_{k,m})\simeq Soc(\R_{k,m})$ since $Rad(\R_{k,m})=I_{u,v}$ and
that $Soc(\R_{k,m})=I_{u^{k-1}v^{m-1}}$. Thus the ring $\R_{k,m}$ is a Frobenius
ring.

\subsection{Linear codes over $\R_{k,m}$}

A linear code $C$ of length $n$ over $\R_{k,m}$ is defined in the usual terms as an $\R_{k,m}$-submodule of $%
\R_{k,m}^{n}$. Define the standard Euclidean inner product on $\R_{k,m}$, that is for $a=(a_{1},a_{2},\ldots
a_{n})$ and  $b=(b_{1},b_{2},\ldots b_{n})\in \R_{k,m}^{n},$ let
\begin{equation*}
\begin{tabular}{l}
$\left\langle a,b\right\rangle =\sum\limits_{i=1}^{n}a_{i}b_{i}.$%
\end{tabular}%
\end{equation*}%
where the operations are performed in the ring $\R_{k,m}.$ The duality for codes over $\R_{k,m}$ then can be defined naturally:
\begin{definition}
Let $C$ be a linear code over $\R_{k,m}$ of length $n$, then we define the
dual of $C$ as
\begin{equation*}
\begin{tabular}{l}
$C^{\perp }:=\left\{ \overline{b}\in \R_{k,m}^{n}\mid \left\langle \overline{b%
},\overline{a}\right\rangle =0,\forall \overline{a}\in C\right\} .$%
\end{tabular}%
\end{equation*}
\end{definition}

\begin{definition}
Let $C$ be a linear code over $\R_{k,m}$ of length $n.$ $C$ is said to be
self-orthogonal if $C\subseteq C^{\perp },$ and self-dual if $C=C^{\perp }$.
\end{definition}

Since $\R_{k,m}$ is a Frobenius ring, by the results in \cite{jwood}, we have
the following lemma:
\begin{lemma}
Any linear code $C$ over $\R_{k,m}^{n}$ satisfies $\left\vert C\right\vert
.\left\vert C^{\perp }\right\vert =\left\vert \R_{k,m}\right\vert ^{n}$
\end{lemma}
A self-dual code will be called Type II if the weights of all codewords are divisible by $4$, otherwise they will be called Type I.

\section{The Lee weight and the Gray map on $\R_{k,m}$}
Our goal in this section is to define a Lee weight for codes over the ring $\R_{k,m}$ and a corresponding Gray map that is distance preserving and more importantly (for the purpose of our work) duality-preserving. In doing so, we will first define these concepts on $\R_{k,1}$ and then inductively extend them over to $\R_{k,m}$.

We define the following linear map which takes a linear code over $%
\R_{k,1}$ of length $n$ to a binary linear code of length $kn$.

\begin{definition}
Take an element $\bar{a}=\overline{a}_{0}+\overline{a}_{1}u+\overline{a}_{2}%
u^{2}+\cdots +\overline{a}_{k-2}u^{k-2}+\overline{a}_{k-1}u^{k-1}$ of $%
(\R_{k,1})^{n}$, where $\overline{a}_{i} \in \F_2^n$. Then define the Gray map $\phi_{k1}$ from $(\R_{k,1})^{n}$
to $(\mathbb{F}_{2})^{kn}$ as follows: when $k$ is even let
\begin{equation*}
\begin{tabular}{ll}
$\phi _{k1}(\bar{a})=$ & $(\overline{a}_{0}+\overline{a}_{1}+\cdots +%
\overline{a}_{k-2}+\overline{a}_{k-1},\overline{a}_{1}+\cdots +\overline{%
a}_{k-2}+\overline{a}_{k-1},$ \\
& $\overline{a}_{1}+\cdots +\overline{a}_{k-2},\cdots ,\overline{a}_{\frac{k}{%
2}-1}+\overline{a}_{\frac{k}{2}}+\overline{a}_{\frac{k}{2}+1},\overline{a}_{%
\frac{k}{2}-1}+\overline{a}_{\frac{k}{2}},\overline{a}_{\frac{k}{2}})$
\end{tabular}%
\end{equation*}%
and when $k$ is odd let
\begin{equation*}
\begin{tabular}{ll}
$\phi_{k1}(\bar{a})=$ & $(\overline{a}_{0}+\overline{a}_{1}+\cdots +%
\overline{a}_{k-2}+\overline{a}_{k-1},\overline{a}_{1}+\cdots +\overline{%
a}_{k-2}+\overline{a}_{k-1},$ \\
& $\overline{a}_{1}+\cdots +\overline{a}_{k-2},\cdots ,\overline{a}_{\frac{k-3%
}{2}}+\overline{a}_{\frac{k-1}{2}}+\overline{a}_{\frac{k+1}{2}},\overline{%
a}_{\frac{k-1}{2}}+\overline{a}_{\frac{k+1}{2}},\overline{a}_{\frac{k-1}{2}}%
).$%
\end{tabular}%
\end{equation*}
\end{definition}

To preserve distance, we define the Lee weight of an element $%
a=a_{0}+a_{1}u+\cdots +a_{k-1}u^{k-1}$ of $\R_{k,1}$ as $w_{L}(a)=w_{H}(\phi
_{k1}(a))$ where $w_{H}$ denotes the usual Hamming weight.

With these definitions, it is obvious that $\phi_{k1}$ is a distance preserving linear isometry from $\R_{k,1}^n$ with the Lee distance to $\F_2^{kn}$ with the Hamming distance. As pointed out earlier, we also want the map to preserve duality, which is proven in the next theorem:

\begin{theorem}
The Gray image of a self-dual code of length $n$ over $\R_{k,1}$ is a binary
self-dual code of length $kn$.
\end{theorem}

\begin{proof}
First, we prove that Gray images of orthogonal codewords in $\R_{k,1}$
are orthogonal in $\mathbb{F}_{2}.$\ That is, we shall show that
\begin{equation*}
\begin{tabular}{l}
$\left\langle \overline{a},\overline{b}\right\rangle =0\Rightarrow \phi_{k1}(%
\overline{a}).\phi_{k1}(\overline{b})=0$%
\end{tabular}%
\end{equation*}%
for all $\overline{a},\overline{b}$ $\in $ $\R_{k,1}^{n}.$ Let us assume
that $\overline{a}=\sum\limits_{i=0}^{k-1}\overline{a}_{i}u^{i}$ and $%
\overline{b}=\sum\limits_{j=0}^{k-1}\overline{b}_{j}u^{j}$. Then we see
that
\begin{equation}\label{eq}
\begin{tabular}{l}
$\left\langle \overline{a},\overline{b}\right\rangle =0\Leftrightarrow
\sum\limits_{i=0}^{k-1}\overline{a}_{i}u^{i}.\sum\limits_{j=0}^{k-1}%
\overline{b}_{j}u^{j}=0\Leftrightarrow \sum\limits_{i+j=0}^{k-1}\overline{%
a}_{i}\overline{b}_{j}=0.$%
\end{tabular}%
\end{equation}%
Now, since%
\begin{equation*}
\begin{tabular}{l}
${\ \phi_{k1}(}\overline{a}{\ )=(}\sum\limits_{i=0}^{k-1}\overline{a}_{i}{\ ,}%
\sum\limits_{i=1}^{k-1}\overline{a}_{i}{\ ,}\sum\limits_{i=1}^{k-2}%
\overline{a}_{i}{\ ,\cdots ,}\sum\limits_{i=\frac{k}{2}-1}^{\frac{k}{2}+1}%
\overline{a}_{i}{\ ,}\sum\limits_{i=\frac{k}{2}-1}^{\frac{k}{2}}\overline{%
a}_{i}{\ ,}\sum\limits_{i=\frac{k}{2}}^{\frac{k}{2}}\overline{a}_{i}{\ ),}$
\\
${\ \phi_{k1} (}\overline{b}{\ )=(}\sum\limits_{i=0}^{k-1}\overline{b}_{i}{\ ,}%
\sum\limits_{i=1}^{k-1}\overline{b}_{i}{\ ,}\sum\limits_{i=1}^{k-2}%
\overline{b}_{i}{\ ,\cdots ,}\sum\limits_{i=\frac{k}{2}-1}^{\frac{k}{2}+1}%
\overline{b}_{i}{\ ,}\sum\limits_{i=\frac{k}{2}-1}^{\frac{k}{2}}\overline{%
b}_{i}{\ ,}\sum\limits_{i=\frac{k}{2}}^{\frac{k}{2}}\overline{b}_{i}{\ )}$%
\end{tabular}%
\end{equation*}%
we get, after some cancellations because of the characteristic being $2$,%
\begin{equation*}
\begin{tabular}{ll}
$\phi _{k1}(\overline{a}).\phi _{k1}(\overline{b})$ & $=\sum%
\limits_{i=0}^{k-1}\overline{a}_{i}\sum\limits_{i=0}^{k-1}\overline{b}_{i}%
+\sum\limits_{i=1}^{k-1}\overline{a}_{i}\sum\limits_{i=1}^{k-1}\overline{%
b}_{i}+\cdots $ \\
& $+\sum\limits_{i=\frac{k}{2}-1}^{\frac{k}{2}}\overline{a}_{i}%
\sum\limits_{i=\frac{k}{2}-1}^{\frac{k}{2}}\overline{b}_{i}+\sum\limits_{i=%
\frac{k}{2}}^{\frac{k}{2}}\overline{a}_{i}\sum\limits_{i=\frac{k}{2}}^{%
\frac{k}{2}}\overline{b}_{i}$ \\
& $=\overline{a}_{0}\sum\limits_{i=0}^{k-1}\overline{b}_{i}+\overline{b}_{0}%
\sum\limits_{i=1}^{k-1}\overline{a}_{i}+\overline{a}_{1}\sum%
\limits_{i=1}^{k-2}\overline{b}_{i}+\overline{b}_{1}\sum\limits_{i=2}^{k-2}%
\overline{a}_{i}+\cdots $ \\
& $+\overline{a}_{\frac{k}{2}-1}\sum\limits_{i=\frac{k}{2}-1}^{\frac{k}{2}}%
\overline{b}_{i}+\overline{b}_{\frac{k}{2}-1}\sum\limits_{i=\frac{k}{2}}^{%
\frac{k}{2}}\overline{a}_{i}.$%
\end{tabular}%
\end{equation*}%
One can see that this last sum is exactly equal to the right-most sum in (\ref{eq}) which is equal to $0$. This shows us
\begin{equation}\label{dualityinclusion}
\begin{tabular}{l}
$\phi _{k1}(C^{\perp })\subset \phi _{k1}(C)^{\perp }.$%
\end{tabular}
\end{equation}
But, by the definition of $\phi _{k1}$, $\phi _{k1}(C)$ is a binary linear
code of length $kn$ of size $\left\vert C\right\vert .$ Both $\mathbb{F}_{2}$
and $\R_{k,1}$ are Frobenius, so we have%
\begin{equation*}
\begin{tabular}{l}
$\left\vert \phi _{k1}(C^{\perp })\right\vert =\left\vert C^{\perp
}\right\vert =\frac{\left\vert \R_{k,1}\right\vert ^{n}}{\left\vert
C\right\vert }=\frac{2^{kn}}{\left\vert \phi _{k1}(C)\right\vert }%
=\left\vert \phi _{k1}(C)^{\perp }\right\vert .$%
\end{tabular}%
\end{equation*}%
Combining this with (\ref{dualityinclusion}), we get
\begin{equation}\label{self-duality}
\phi _{k1}(C^{\perp })=\phi _{k1}(C)^{\perp }.
\end{equation}
\end{proof}

Because of the distance-preserving property of the Gray map we get the following important corollary:
\begin{corollary}\label{dualweight}
Let $C$ be a self-dual code over $\R_{k,1}$ of length $n$. Then $\phi_{k1}(C)$ is a binary self-dual code of length $kn$. Moreover the Lee weight distribution of $C$ is the same as the Hamming weight distribution of $\phi_{k1}(C)$.
\end{corollary}

Now since $\R_{k,m}$ can be viewed as an $\R_{k,1}-$vector space with a basis $\left\{
1,v,v^{2},\ldots ,v^{m-1}\right\}$, we can write any element of $%
\R_{k,m}$ in the form $c=\sum\limits_{0\leq i\leq m-1}c_{ki}v^{i},$ where $%
c_{ki}\in \R_{k,1}.$ Now we can extend the gray map easily from $\R_{k,1}$
to $\R_{k,m}$:
\begin{equation*}
\begin{tabular}{ll}
$\phi_{km}(c)=$ & $\large{(} \phi_{k1}(\sum\limits_{i=0}^{m-1}\overline{c}_{ki}%
),\phi_{k1}(\sum\limits_{i=1}^{m-1}\overline{c}_{ki}),\phi
_{k1}(\sum\limits_{i=1}^{m-2}\overline{c}_{ki}),$ \\
& $\cdots, \phi_{k1}(\sum\limits_{i=\frac{m}{2}-1}^{\frac{m}{2}+1}%
\overline{c}_{ki}),\phi_{k1}(\sum\limits_{i=\frac{m}{2}-1}^{\frac{m}{2}}%
\overline{c}_{ki}),\phi_{k1}(\sum\limits_{i=\frac{m}{2}}^{\frac{m}{2}}%
\overline{c}_{ki}) \large{)}.$%
\end{tabular}%
\end{equation*}%
We note that the, defining the Lee weight in the same way as the Hamming weight of the image, distance and duality-preserving properties of $\phi_{km}$ can be established in exactly the same way as was done for $\phi_{k1}$. Thus we can extend corollary \ref{dualweight} to the following important theorem which will be used in subsequent sections:
\begin{theorem}\label{dual}
Let $C$ be a self-dual code over $\R_{k,m}$ of length $n$. Then $\phi_{km}(C)$ is a binary self-dual code of length $kmn$. Moreover the Lee weight distribution of $C$ is the same as the Hamming weight distribution of $\phi_{km}(C)$.
\end{theorem}

\section{MacWilliams identities for codes over $R_{k,m}$}

MacWilliams identities give a relation between weight enumerators of a code
and its dual. By Jay Wood's result \cite{jwood}, MacWilliams identities hold for codes over all Frobenius rings. Since $\R_{k,m}$ is Frobenius ring it has a generating character and using this we can prove MacWilliams identities for the complete weight
enumerator, the Hamming weight enumerator and the Lee weight enumerator of codes
over $\R_{k,m}.$

We first give a generating character for $\R_{k,m}$. Let
\begin{equation*}
\begin{tabular}{llll}
$\chi :$ & $(\R_{k,m},+)$ & $\rightarrow $ & $(\left\{ -1,1\right\} ,.)$ \\
& $\sum\limits_{\substack{ 0\leq i\leq k-1  \\ 0\leq j\leq m-1}}%
c_{ij}u^{i}v^{j}$ & $\mapsto $ & $(-1)^{w_{H}(c)}$%
\end{tabular}%
,
\end{equation*}%
where $c=(c_{ij})$ is the vector consisting of all the coefficients $c_{ij}$'s. It is clear that $\chi$ is a character.
\begin{theorem}
$\chi $ is generating character for $\R_{k,m}.$
\end{theorem}

\begin{proof}
Since $\chi (0)=1$ and $\chi (u^{k-1}v^{m-1})=-1,$ $\chi $ is
non-trivial when restricted to the minimal ideal. Since every non-zero ideal contains the minimal ideal, $\chi $ is when restricted to any non-zero ideal.
\end{proof}

Let $\R_{k,m}=\left\{ g_{1},g_{2},\ldots ,g_{2^{km}}\right\} $ be the ring.
The complete weight enumerator of a code $C$ over $\R_{k,m}^{n}$ is
\begin{equation*}
\begin{tabular}{l}
$cwe_{C}(\overline{X})=\sum\limits_{\overline{c}\in
C}\prod\limits_{i=1}^{2^{km}}X_{i}^{n_{i}(\overline{c})}$,%
\end{tabular}%
\end{equation*}%
where $n_{i}(\overline{c})$ is the number of occurrences of $g_{i}$ in $%
\overline{c}.$ Let $T$ be the $2^{km}\times 2^{km}$ matrix such that
\begin{equation*}
\begin{tabular}{l}
$T=%
\begin{pmatrix}
\chi (g_{1}g_{1}) & \chi (g_{1}g_{2}) & \cdots & \chi (g_{1}g_{2^{km}}) \\
\chi (g_{2}g_{1}) & \ddots &  & \chi (g_{1}g_{2^{km}}) \\
\vdots &  & \ddots & \vdots \\
\chi (g_{2^{km}}g_{1}) & \chi (g_{2^{km}}g_{2}) & \cdots & \chi
(g_{2^{km}}g_{2^{km}})%
\end{pmatrix}%
.$%
\end{tabular}%
\end{equation*}%
Then we have following theorems by \cite{jwood}:

\begin{theorem}
Let $C$ be linear code over $\R_{k,m}$ and $C^{\perp }$ be its dual. Then we
have the following identity for the complete weight enumerators:%
\begin{equation*}
\begin{tabular}{l}
$cwe_{C^{\perp }}(\overline{X})=\frac{1}{\left\vert C\right\vert }cwe_{C}(T.%
\overline{X}^{t}).$%
\end{tabular}%
\end{equation*}%
Here, $\overline{X}^{t}$ denotes the transpose of $\overline{X}.$
\end{theorem}

Putting $X_1=x$ and $X_i=y$ for all $i \geq 2$, we obtain the MacWilliams identity for the Hamming weight enumerator:
\begin{theorem}
\begin{equation*}
\begin{tabular}{l}
$W_{C^{\perp }}(x,y)=\frac{1}{\left\vert C\right\vert }W_{C}(x+(\left\vert
\R_{k,m}\right\vert -1)y,x-y)$,
\end{tabular}%
\end{equation*}%
where $W_{C}(x,y)\ $is Hamming weight enumerator of a code $C$ over $%
\R_{k,m}^{n}$ in the usual way,%
\begin{equation*}
\begin{tabular}{l}
$W_{C}(x,y)=\sum\limits_{c\in C}x^{n-w_{H}(c)}y^{w_{H}(c)}.$%
\end{tabular}%
\end{equation*}
\end{theorem}

Now, our goal is to describe MacWilliams identities for the Lee weight
enumerators of codes over $\R_{k,m}.$ Firstly, we define Lee weight
enumerator of a code $C$ over $\R_{k,m}^{n}$ as usual to be
\begin{equation*}
\begin{tabular}{l}
$Lee_{C}(z)=\sum\limits_{\overline{c}\in C}z^{w_{L}(\overline{c})}$%
\end{tabular}%
\end{equation*}%
where $w_{L}(\overline{c})$ denotes the Lee weight of a codeword. Then we
have following theorem:

\begin{theorem}
Let $C$ be a linear code over $\R_{k,m}$ of length $n$ then
\begin{equation*}
\text{%
\begin{tabular}{l}
$Lee_{C^{\perp }}(z)=\frac{1}{\left\vert C\right\vert }(1+z)^{kmn}Lee_{C}%
\left( \frac{1-z}{1+z}\right) .$%
\end{tabular}%
}
\end{equation*}
\end{theorem}

\begin{proof}
As we know $\phi_{km}$ is a distance preserving map. Therefore
\begin{equation*}
\begin{tabular}{l}
$Lee_{C^{\perp }}(z)=W_{\phi_{km}(C^{\perp })}(z)$
\end{tabular}%
\end{equation*}%
where $W_C(z)$ denotes the hamming weight enumerator of a code $C$.
Recall that we have $\phi_{km}(C^{\perp })=\phi_{km}(C)^{\perp }$ by
Theorem \ref{self-duality}. So we get%
\begin{equation*}
\begin{tabular}{ll}
$Lee_{C^{\perp }}(z)$ & $=W_{\phi _{km}(C)^{\perp }}(z)$ \\
& $=\frac{1}{\left\vert \phi _{km}(C)\right\vert }(1+z)^{kmn}W_{\phi
_{km}(C)}\left( \frac{1-z}{1+z}\right) $ \\
& $=\frac{1}{\left\vert C\right\vert }(1+z)^{kmn}Lee_{C}\left( \frac{1-z}{1+z%
}\right) .$%
\end{tabular}%
\end{equation*}
\end{proof}

\section{Projections, Lifts and Constructions of Self-Dual Codes Over $R_{k,m}$}

\subsection{Projections and Lifts}

Recall that elements of $\R_{k,m}$ can be shown in the form $\sum\limits
_{\substack{ 0\leq i\leq k-1  \\ 0\leq j\leq m-1}}c_{ij}u^{i}v^{j}.$ Now
define a projection of $\R_{k,m}$ to $\mathbb{F}_{2}.$

\begin{definition}
Let $\mu $ be a map from $\R_{k,m}$ to $\mathbb{F}_{2}$ such that%
\begin{equation*}
\begin{tabular}{l}
$\mu (\sum\limits_{\substack{ 0\leq i\leq k-1  \\ 0\leq j\leq m-1}}%
c_{ij}u^{i}v^{j})=c_{00}$%
\end{tabular}%
\end{equation*}%
Then $\mu $ is an epimorphism and is called a natural projection of $\R_{k,m}$
to $\mathbb{F}_{2}.$
\end{definition}

Let $C$ be a linear code over $\R_{k,m}$ and $\mu (C)$ be its projection.
Then $C$ is said to be a lift of $\mu (C).$ Our general strategy in constructing self-dual codes over $\R_{k,m}$ will be to lift from good binary self-dual codes. Now notice that if for $\overline{x}, \overline{y} \in \R_{k,m}^n$, we have $\langle \overline{x}, \overline{y} \rangle = 0$, then $\overline{x}_{00}\cdot \overline{y}_{00} = \mu(\overline{x})\cdot \mu(\overline{y}) = 0$. Thus we have the following result:
\begin{theorem}
Let $C$ be a self-dual code over $\R_{k,m}$ of length $n$. Then $\mu(C)$ is
a self orthogonal code over $\mathbb{F}_{2}$ of length $n.$
\end{theorem}

\begin{corollary}\label{projselfdual}
If $C$ is a free self-dual code over $\R_{k,m}$ of length $2n$, that is $C$ is generated by a matrix of the form $[I_n|A]$, then $\mu(C)$ is a binary self-dual code of length $2n$.
\end{corollary}

The following theorem gives a bound between the minimum Lee weight of a code and the minimum Hamming weight of its projection:

\begin{theorem}\label{bound}
Let $C$ be a linear code over $\R_{k,m}$ of length $n$ with minimum Lee
weight $d$ and $\mu (C)$ be its projection to $\mathbb{F}_{2}$. If $%
d^{\prime }$denotes the minimum Hamming weight of $\mu(C),$ we have $d\leq
2md^{\prime }.$
\end{theorem}

\begin{proof}
Let $\overline{x}_{00}\in \mu (C)$ with $w_{H}(\overline{x}_{00})=d^{\prime
}.$ Then there exists $c=\overline{x}_{00}+\sum\limits_{1\leq i+j\leq k+m-2}%
\overline{x}_{ij}u^{i}v^{j}\in C.$ But then $(u^{k-1}v^{m-1})c=\overline{x_{00}}%
u^{k-1}v^{m-1}\in C,$ because $C$ is linear code over $\R_{k,m}.$ Now,
\begin{equation*}
\begin{tabular}{l}
$w_{L}(\overline{x}_{00}u^{k-1}v^{m-1})=w_{H}(\underset{m\ \text{times }%
\overline{x}_{00},\overline{x}_{00},\overline{00}}{\underbrace{\overline{x}_{00},\overline{x}_{00},\overline{00},
\overline{x}_{00},\overline{x}_{00},
\overline{00},\cdots ,\overline{x}_{00},\overline{x}_{00},\overline{00}}})$%
\end{tabular}%
\end{equation*}%
where $\overline{00}=\underset{k-2\text{ times}}{\underbrace{\overline{0}%
,\cdots ,\overline{0}}}.$ That is, $w_{L}(u^{k-1}v^{m-1}\overline{x}_{00}%
)=2md^{\prime }$. This proves the theorem.
\end{proof}

\subsection{Self-Dual Codes Over $R_{k,m}$}

The double circulant and bordered double circulant constructions described in \cite{MacWilliams} have been used quite successfully by many researchers to obtain good self-dual binary codes. We can easily adopt these constructions to $\R{k,m}$:

\begin{definition}
Let $M$ be a circulant matrix over $\R_{k,m}$ of order $n$. Then the matrix $%
\left[ I_{n}\mid M\right] $ generates codes over $\R_{k,m}$ of length $2n.$
This is called the pure double circulant or double circulant construction.
\end{definition}

\begin{definition}
If $M$ be a circulant matrix over $\R_{k,m}$ of order $n-1.$ Then the matrix
\begin{equation*}
\begin{tabular}{l}
$\left[ I_{n}\left\vert
\begin{tabular}{llll}
$x$ & $y$ & $\cdots $ & $y$ \\
$z$ &  &  &  \\
$\vdots $ &  & $M$ &  \\
$z$ &  &  &
\end{tabular}%
\right. \right] $%
\end{tabular}%
\end{equation*}%
where $x,y,z\in \R_{k,m}$ generates codes over $\R_{k,m}$ of length $2n.$ This
is called bordered double circulant construction.
\end{definition}

Another construction, which is more recent was given in \cite{betsumiya} for self-dual codes over $%
\mathbb{F}_{p}.$ In \cite{georgiou} it was called two-block circulant
construction and later it was called the four circulant
construction. In \cite{karadenizyildizaydin} this construction was applied to the ring $\mathbb{F}_{2}+u\mathbb{F}_{2}
$ to obtain extremal binary self-dual codes.
Then following theorem can be proven in the exact same way as was done in \cite{karadenizyildizaydin}:
\begin{theorem}
Let $A$ and $B$ be circulant matrix over $\R_{k,m}$ of length $n$ such that $%
AA^{t}+BB^{t}=I_{n}.$ Then the matrix
\begin{equation*}
\begin{tabular}{l}
$\left[ I_{2n}\left\vert
\begin{array}{cc}
A & B \\
B^{t} & A^{t}%
\end{array}%
\right. \right] $%
\end{tabular}%
\end{equation*}%
generates self dual codes over $\R_{k,m}$ of length $4n.$ This is called
four-circulant construction.
\end{theorem}

Now, we can give self-dual codes over $\mathbb{F}_{2}$ of some length
obtained from self-dual codes over $\R_{k,m}$ by the three constructions mentioned above, using the Magma computer algebra system (\cite{Magma}).

\subsubsection{The General idea}

The projection $\mu $ which is defined above preserves orthogonality. Also
the image of a double circulant self dual code over $R_{k,m}$ of length $n$
under $\mu $ must be a double circulant binary self dual code, the image of
a bordered-double circulant self dual code over $R_{k,m}$ of length $n$ under $%
\mu $ has to be a bordered-double circulant binary self dual code and the same is true for four circulant codes as well.

So, if we want to obtain a good self-dual code over  $\R_{k,m}$ by one of the construction methods above, we look at the projection and look for the best binary self-dual codes of the same length obtained from the same constructions. We then lift these codes over the ring $\R_{k,m}$ by taking lifting $1$ to a unit in $\R_{k,m}$ and $0$ to a non-unit in $\R_{k,m}$. Theorem \ref{bound} tells us exactly which binary codes to lift. Then an exhaustive search using a computer algebra reveals all the self-dual codes over $\R_{k,m}$ that can be obtained through these constructions. We then choose the best ones and take the Gray images to obtain good binary self-dual codes. In what follows we apply this idea to certain lengths and certain rings of the form $\R_{k,m}$. We only list the ones through which we have obtained extremal or near extremal binary self-dual codes.

Recall that for binary self-dual codes we have the following upper bounds on the minimum
Hamming distance:

\begin{theorem}
$($\cite{conwaysloane}$)$ Let $d_{I}(n)$ and $d_{II}(n)$ be the minimum distance of a
Type I and Type II binary code of length $n$, respectively. Then
\begin{equation*}
d_{II}(n)\leq 4\lfloor \frac{n}{24}\rfloor +4
\end{equation*}%
and
\begin{equation*}
d_{I}(n)\leq \left\{
\begin{array}{ll}
4\lfloor \frac{n}{24}\rfloor +4 & \text{if $n\not\equiv 22\pmod{24}$} \\
4\lfloor \frac{n}{24}\rfloor +6 & \text{if $n\equiv 22\pmod{24}$.}%
\end{array}%
\right.
\end{equation*}
\end{theorem}

Self-dual codes meeting these bounds are called \textit{extremal}.
The existence of the Type II extremal code of length $72$ is still an open problem. So the best known binary self-dual codes of length $72$ for both Type I and Type II have parameters $[72,36,12]$.

\subsection{The extended binary Golay code}
The binary Golay code is probably the most well known code in the literature. It is a perfect $3$-error correcting code of parameters $[23,12,7]$. When we extend this code by a parity check symbol we obtain the Type II extremal self-dual code of parameters $[24,12,8]$. This code is unique up to equivalence and is the first example of the theoretically good self-dual codes of length $24k$. Using Assmus-Mattson theorem, it also leads designs with good parameters. There have been many different constructions for this code in the literature. \cite{Hurley}, \cite{Peng} are examples of these constructions. In \cite{Golay}, the extended Golay code was constructed from what we now call $\R_{2,2}$.

We have been able to give a construction for the extended Golay code using bordered double circulant construction over $\R_{3,1}$ and $\R_{3,2}$. Note that because of the Gray map, these are the only ones we can use (other than $\R_{2,1}$ and $\R_{2,2}$, which have already been used before). To construct it from $\R_{3,1}$, we need the binary code to lift to be of parameters $[8,4,4]$ which is also unique. All possible lifts of the bordered double circulant matrix that generates the $[8,4,4]$-code we were able to obtain the Golay code from $\R_{3,1}$ quite easily.
The following matrix turns out to generate the self-dual code over $\R_{3,1}$ whose binary image is the extended Golay code:
$$M = \left [ \begin{array}{cccccccc} 1 & 0 & 0 & 0 &u+u^2 &1+u &1+u &1+u \\
0 & 1 & 0 & 0 &1+u &u &1 &1+u^2 \\
0 & 0 & 1 & 0 &1+u &1+u^2 &u &1 \\
0 & 0 & 0 & 1 &1+u &1 &1+u^2 &u
\end{array} \right ]. $$

Doing the same thing over bordered double circulant binary codes of length $4$, which narrowed the search field rather considerably, we see that the following matrix generates the self-dual code over $\R_{3,2}$ whose binary image is the extended Golay code:

$$M' = \left [ \begin{array}{cccc} 1 & 0 & u+v & 1+u+v \\
0 & 1 & 1+u+u^2+v+uv & u+v
\end{array} \right ]. $$

\subsection{Extremal Self Dual Codes of Length 36}

Melchor and Gaborit have classified all the $41$ extremal binary $[36,18,8]$
self-dual codes in \cite{melchorgaborit}. We have obtained some of these
through $\R_{3,1}$ and $\R_{3,2}$ using some of the aforementioned constructions. To be precise, we found $6$ of the $41$ extremal self-dual codes from the constructions mentioned above. Now, since the four-circulant codes have to be of length divisible by $4$, the four circulant construction was applied only to the case of $\R_{3,1}$, whereas the double circulant and the bordered double circulant constructions were applied to both $\R_{3,1}$ and $\R_{3,2}$. In the case of $\R_{3,1}$ we searched for all the good binary self-dual codes of length $12$ (in this case with the parameters [12,6,4]) and then lifted them. In the case of $\R_{3,2}$ we lifted all the good binary self-dual codes of length $6$.

After searching over all possible lifts that are self-dual and taking Gray
images of these lifts we have obtained $6$ non-equivalent extremal self-dual
codes of length $36$. Two of these codes also have been obtained taking Gray
images of double circulant self-dual codes over $\R_{3,1}$ and $\R_{3,2}$ of
length $12.$

There are two weight enumerators are possible by \cite{conwaysloane}:%
\begin{equation*}
\begin{tabular}{l}
$W_{36,1}=1+225y^{8}+2016y^{10}+\cdots $%
\end{tabular}%
\end{equation*}%
and
\begin{equation*}
\begin{tabular}{l}
$W_{36,2}=1+289y^{8}+1632y^{10}+\cdots $%
\end{tabular}%
\end{equation*}%

\begin{table}[H]
\caption{binary [36,18,8] extremal self-dual codes obtained from double circulant constructions}
\begin{tabular}{|c|c|c|c|}\hline
Ring & First row of $M$ & $\left\vert Aut(C)\right\vert $ & $W_{36}(C)$ \\
\hline
$\R_{3,1}$ & $\left( u^{2}+u,1,u+1,u^{2}+u+1,u^{2}+u+1,1\right) $ & $864$ & $%
W_{36,1}$ \\ \hline
$\R_{3,2}$ & $\left( u+v,u^{2}+u+v,u^{2}v+uv+v+1\right) $ & $864$ & $W_{36,1}$
\\ \hline
$\R_{3,1}$ & $\left( u,1,u+1,u^{2}+u+1,u^{2}+u+1,1\right) $ & $12960$ & $%
W_{36,1}$ \\ \hline
$\R_{3,2}$ & $\left( u+v,u^{2}v+u^{2}+u+v,uv+1\right) $ & $12960$ & $W_{36,1}$
\\ \hline
\end{tabular}
\end{table}

\begin{table}[H]
\caption{binary [36,18,8] extremal self-dual codes obtained from bordered double circulant construction over $\R_{3,1}$}
\begin{tabular}{|c|c|c|c|}\hline
First row of $M$ & $(x,y,z)$ & $\left\vert Aut(C)\right\vert $ & $W_{36}(C)$
\\ \hline
$(u,1,1,u^{2}+1,u^{2}+1)$ & $(u,u+1,u+1)$ & $80$ & $W_{36,2}$ \\ \hline
$(u,1,u+1,u^{2}+u+1,1)$ & $(u^{2}+u,u+1,u+1)$ & $240$ & $W_{36,1}$ \\ \hline
\end{tabular}
\end{table}

\begin{table}[H]
\caption{binary [36,18,8] extremal self-dual codes obtained from four circulant construction over $\R_{3,1}$}
\begin{tabular}{|c|c|c|c|}\hline
First row of $A$ & First row of $B$ & $\left\vert Aut(C)\right\vert $ & $%
W_{36}(C)$ \\ \hline
$\left( u,1,u^{2}+1\right) $ & $\left( u+1,u+1,u+1\right) $ & $96$ & $%
W_{36,1}$ \\ \hline
$\left( u^{2}+u,1,u^{2}+1\right) $ & $\left( u+1,u+1,u+1\right) $ & $288$ & $%
W_{36,1}$ \\ \hline
$\left( u,1,u^{2}+1\right) $ & $\left( u+1,u+1,u^{2}+u+1\right) $ & $864$ & $%
W_{36,1}$ \\ \hline
$\left( u^{2}+u,1,u^{2}+1\right) $ & $\left( u+1,u+1,u^{2}+u+1\right) $ & $%
12960$ & $W_{36,1}$ \\ \hline
\end{tabular}
\end{table}

\subsection{Extremal Self Dual Codes of Length 66}
Extremal codes f length $66$ have parameter $[66,33,12]$ and their possible weight enumerators are as follows:
\begin{eqnarray*}
W_{66,1} &=&1+\left( 858+8\beta \right) y^{12}+\left( 18678-24\beta \right)
y^{14}+\cdots \text{ where }0\leq \beta \leq 778, \\
W_{66,2} &=&1+1690y^{12}+7990y^{14}+\cdots \text{ } \\
\text{and }W_{66,3} &=&1+\left( 858+8\beta \right) y^{12}+\left(
18166-24\beta \right) y^{14}+\cdots \text{ where }14\leq \beta \leq 756.
\end{eqnarray*}%

We have obtained $2$ non-equivalent extremal self-dual $[66,33,12]_{2}$
codes from double circulant matrices over $\R_{3,1}$. Because of Theorem \ref{bound}, we needed to search for the $[22,11,6]$ binary double circulant self-dual code, which we lifted to $\R_{3,1}$. After taking Gray images of these lifts we have obtained
the following extremal self-dual $[66,33,12]_{2}$ codes, which were also obtained in \cite{akaya} by a different construction:

\begin{table}[H]
\caption{binary [66,33,12] extremal self-dual codes obtained from double circulant construction over $\R_{3,1}$}
\begin{tabular}{|c|c|c|}\hline
First row of $A$ & $\left\vert Aut(C)\right\vert $ & %
$\beta$ in $W_{66,1}$ \\ \hline
$\left( u,u,u,1,u,u^{2}+u,1,u,1,1,1\right) $  & $220$& $22$\\ \hline
$\left( u^{2}+u,u^{2}+u,u^{2}+u,1,u^{2}+u,u,1,u^{2}+u,1,1,1\right)$  & $660$& $66$\\ \hline
\end{tabular}
\end{table}

\subsection{Best known Self-dual Codes of Length 72}
We know that an extremal Type I code of length $72$ must have a minimum distance $14$ while a Type II one must have $16$
as its minimum distance. But as yet the existence of these codes is an open problem. However a lot of work has gone towards classifying selþf-dual codes of  parameters $[72,36,12]$ of both types, especially Type II ones.

A number of singly even self-dual $[72,36,12]_{2}$ codes have been listed in \cite{akaya} and \cite{doughertykim}.\ In \cite{gulliverharada}, \cite{doughertygulliverharada}, \cite{rdontcheva}, \cite{bouyukliev} a great number of doubly even self-dual $[72,36,12]_{2}$ codes
are constructed.

We have constructed a lot of new Type I and Type II self-dual codes of length $72$ as images of self-dual codes over $\R_{3,1}$ and $\R_{3,2}$
via the double and bordered double circulant constructions. To do this, by using Theorem \ref{bound}, we have had to do an exhaustive search over all possible lifts of suitable binary self-dual codes of length $24$ or $12$. Since Type II codes are of more importance in the literature, and to save space, we have not listed all the Type I codes. But we give the parameters of the ones we have found and we have put the generators in a database that can be reached at.

Using the double circulant construction over $\R_{3,1}$ we were able to obtain
$117$ non-equivalent Type \textrm{I} binary $%
[72,36,12]$-codes and $43$ non-equivalent Type \textrm{%
II} binary $[72,36,12]$-codes. Using the bordered double circulant construction over $\R_{3,1}$ we found $27$ new Type II
and $36$ Type I self-dual $[72,36,12]$-codes. Moreover, by using the bordered double circulant construction over $\R_{3,2},$ we constructed $22$ new Type \textrm{%
I} $[72,36,12]$-codes and $35$ new Type \textrm{II} $[72,36,12]$-codes.

 In \cite{akaya}
two possible weight enumerators were given for Type \textrm{I} $%
[72,36,12]$-codes as follows:
\begin{equation*}
\begin{tabular}{l}
$W_{72,1}=1+2\beta y^{12}+(8640-64\gamma )y^{14}+(124281-24\beta +384\gamma
)y^{16}+$\textperiodcentered \textperiodcentered \textperiodcentered \\
$W_{72,2}=1+2\beta y^{12}+(7616-64\gamma )y^{14}+(134521-24\beta +384\gamma
)y^{16}+$\textperiodcentered \textperiodcentered \textperiodcentered%
\end{tabular}%
\end{equation*}%
where $\beta $ and $\gamma $ are parameters. The possible weight enumerators for a Type \textrm{II} $[72,36,12]$ code are given in \cite{doughertygulliverharada} as
\begin{equation*}
\begin{tabular}{l}
$W_{72}=1+(4398+\alpha )y^{12}+(197073-12\alpha )y^{16}+$\textperiodcentered
\textperiodcentered \textperiodcentered%
\end{tabular}%
\end{equation*}

The $76$ new binary Type \textrm{I }$[72,36,12]$ self-dual codes that were obtained
from double circulant matrices over $\R_{3,1}$ have all $48$ as the order of their automorphism group and the parameters for their weight enumerators in $W_{72,1}$ are $\gamma =$0 and $\beta =$185,\ 199,\ 201,\ 207,\ 225,\ 231,\
233,\ 247,\ 249,\ 255,\ 271,\ 273,\ 281,\ 295,\ 297,\ 303,\ 317,\ 319,\
321,\ 329,\ 339,\ 341,\ 353,\ 355,\ 375,\ 377,\ 463$,\ \gamma =$6 and $\beta
=$145,\ 153,\ 159,\ 165,\ 169,\ 171,\ 177,\ 181,\ 183,\ 193,\ 195,\ 201,\
207,\ 213,\ 217,\ 219,\ 225,\ 231,\ 237,\ 243,\ 253,\ 255,\ 265,\ 267,\
277,\ 279,\ 285,\ 291,\ 297,\ 303,\ 309,\ 315$,\ $321,\ 325,\ 327,\ 345,\
349,\ 351,\ 387,\ 411,\ 423$,\ \gamma =$24 and $\beta =$345,\ 393,\ 411,\
427,\ 429,\ 449,\ 453,\ 497$. $ 41 new binary Type \textrm{I }$[72,36,12]$
self-dual codes that were obtained from double circulant matrices over $%
\R_{3,1}$ had all $96$ as the order of their automorphism group and their parameters in $W_{72,1}$ were $\gamma =$0 and $\beta =$%
235,\ 259,\ 291,\ 315,\ 331,\ 339,\ 341,\ 355,\ 357,\ 363,\ 365,\ 379,\
381,\ 389,\ 403,\ 413,\ 427,\ 429,\ 435,\ 459,\ 485,\ 499,\ 507,\ 509$,$ $%
\gamma =$24 and $\beta =$331,\ 333,\ 339,\ 355,\ 357,\ 363,\ 381,\ 411,\
427,\ 453,\ 483,\ 499,\ 501,\ 525,\ 573$,\ \gamma =$48 and $\beta =$629,\ 653%
$.$

The order of the automorphism group of all the $36$ new binary Type \textrm{I }$[72,36,12]$ self-dual codes that were constructed from bordered-double circulant matrices over $\R_{3,1}$ is $44$ and their parameters in $W_{72,2}$ are given as $\gamma =$0 and $\beta =$88,\ 89,\ 111,\ 132,154,\
155,\ 165,\ 177,\ 187,\ 198,\ 199,220,\ 221,\ 231,\ 242,\ 243,\ 253,\ 264,\
265,\ 275,\ 286,\ 287,\ 297,\ 309,\ 319,\ 330,\ 331,\ 353,\ 363,\ 374,\
385,\ 397$,$\ 418,\ 462,\ 573,\ 617$.$

The $22$ new binary Type \textrm{I }$[72,36,12]$ self-dual codes that were obtained
from bordered-double circulant matrices over $\R_{3,2}$ had weight
enumerator of the form $W_{72,1}.$ The codes whose automorphism groups are of order $40$ have
$\gamma =$24 and $\beta =$363,\ 383$,\ \gamma =$12 and $\beta =$319,\ 359,\
379$,\ \gamma =$10 and $\beta =$249,\ 269,\ 289,\ 309$.$ in $W_{72,1}$. The codes whose automorphism groups are of order $20$ have $\gamma =$18 and $\beta =$176,\ 186,\
206,\ 236,\ 246,\ 276,\ 296, $\gamma =$16 and $\beta =$197,\ 237,\ 277$,\
\gamma =$9 and $\beta =$267,\ 287,\ 307 in $W_{72,1}$.

\begin{remark}
Due to their relative importance in the literature we only list the constructions of Type II self-dual codes of length $72$. However the reader can find the constructions for all the Type I new codes of length $72$ that are mentioned above, in the database \cite{tufekci}
\end{remark}

Before proceeding with the following tables in which we list all the new Type II binary self-dual codes of parameter $[72,36,12]$, we would like to introduce a notation to shorten the elements of $\R_{3,2}$, that can also be used for $\R_{3,1}$ as well. Note that $R_{3,2}$ is an $\F_2$-vector space with a basis that we can take as $\{u^2v, uv, v, u^2, u, 1\}$. Any element in $\R_{3,2}$ corresponds to a 6-bit string over $\F_2$ which we can consider as a base 2 expression of a natural number. With this notation every element in $\R_{3,2}$ corresponds to a integer from $0$ to $63$. For example $uv+v+u^2+1$ corresponds to $(011101)$ whose numerical value can be taken as $29$. Taking the basis as $\{u^2,u,1\}$ gives a numerical value from $0$ to $7$ to any element in $\R_{3,1}$.

\newpage
\begin{table}[H]
\caption{New Type II [72,36,12]
self-dual codes obtained from double circulant matrices over $\R_{3,1}$}
\begin{tabular}{|c|c|c|c|}
\hline
code $C_i$ & first row of $M$ & $\alpha$ in $W_{72}$  & $|Aut(C_i)|$ \\ \hline
$C_{1}$ & $\left( 2,0,4,3,6,1,3,3,5,4,7,5\right) $ & -3996 & 144 \\ \hline
$C_{2}$ & $\left( 0,6,0,3,6,3,3,7,1,6,5,7\right) $ & -3900 & 48\\ \hline
$C_{3}$ & $\left( 0,0,6,1,2,3,3,5,3,0,1,7\right) $ & -3888 & 48\\ \hline
$C_{4}$ & $\left( 0,6,4,3,2,3,3,7,1,6,1,7\right) $ & -3876 & 48\\ \hline
$C_{5}$ & $\left( 0,0,2,1,2,1,3,5,7,0,1,5\right) $ & -3852 & 48\\ \hline
$C_{6}$ & $\left( 2,0,2,3,6,7,3,3,3,4,7,3\right) $ & -3804 & 48\\ \hline
$C_{7}$ & $\left( 0,0,6,3,4,5,1,3,7,4,5,1\right) $ & -3768 & 48\\ \hline
$C_{8}$ & $\left( 0,0,0,1,6,3,3,5,1,0,5,7\right) $ & -3756 & 48\\ \hline
$C_{9}$ & $\left( 0,6,0,3,2,1,3,7,5,6,1,5\right) $ & -3744 & 48\\ \hline
$C_{10}$ & $\left( 0,4,4,1,2,3,3,5,1,4,1,7\right) $ & -3732 & 48\\ \hline
$C_{11}$ & $\left( 0,6,2,3,2,1,3,7,7,6,1,5\right) $ & -3708 & 48\\ \hline
$C_{12}$ & $\left( 2,0,0,3,6,3,3,3,1,4,7,7\right) $ & -3696 & 48\\ \hline
$C_{13}$ & $\left( 0,0,6,3,4,1,1,3,7,4,5,5\right) $ & -3672 & 48\\ \hline
$C_{14}$ & $\left( 0,0,4,3,4,5,1,3,5,4,5,1\right) $ & -3660 & 48\\ \hline
$C_{15}$ & $\left( 0,0,4,1,6,1,3,5,5,0,5,5\right) $ & -3624 & 48\\ \hline
$C_{16}$ & $\left( 2,0,2,3,6,3,3,3,3,4,7,7\right) $ & -3612 & 48\\ \hline
$C_{17}$ & $\left( 0,0,0,3,4,7,1,3,1,4,5,3\right) $ & -3600 & 7920\\ \hline
$C_{18}$ & $\left( 0,0,6,1,2,7,3,5,3,0,1,3\right) $ & -3600 & 48\\ \hline
$C_{19}$ & $\left( 0,0,4,1,2,3,3,5,1,0,1,7\right) $ & -3588 & 48\\ \hline
$C_{20}$ & $\left( 0,0,2,1,2,5,3,5,7,0,1,1\right) $ & -3564 & 48\\ \hline
$C_{21}$ & $\left( 0,4,2,1,2,5,3,5,7,4,1,1\right) $ & -3564 & 144\\ \hline
$C_{22}$ & $\left( 0,0,2,3,0,5,1,3,7,4,1,1\right) $ & -3552 & 48\\ \hline
$C_{23}$ & $\left( 0,0,2,3,4,7,1,3,3,4,5,3\right) $ & -3516 & 48\\ \hline
$C_{24}$ & $\left( 0,0,0,3,0,1,1,3,5,4,1,5\right) $ & -3492 & 48\\ \hline
$C_{25}$ & $\left( 0,0,2,1,6,3,3,5,3,0,5,7\right) $ & -3480 & 48\\ \hline
$C_{26}$ & $\left( 0,2,2,1,4,7,1,1,3,6,5,3\right) $ & -3468 & 48\\ \hline
$C_{27}$ & $\left( 0,0,0,1,2,1,3,5,5,0,1,5\right) $ & -3456 & 48\\ \hline
$C_{28}$ & $\left( 0,4,6,1,6,1,3,5,7,4,5,5\right) $ & -3444 & 48\\ \hline
$C_{29}$ & $\left( 2,0,4,3,2,7,3,3,1,4,3,3\right) $ & -3384 & 48\\ \hline
$C_{30}$ & $\left( 2,0,4,1,4,1,1,5,5,0,7,5\right) $ & -3336 & 48\\ \hline
$C_{31}$ & $\left( 0,2,6,3,2,7,3,7,3,2,1,3\right) $ & -3312 & 48\\ \hline
$C_{32}$ & $\left( 0,0,0,3,0,5,1,3,5,4,1,1\right) $ & -3300 & 48\\ \hline
$C_{33}$ & $\left( 0,0,0,3,4,3,1,3,1,4,5,7\right) $ & -3264 & 48\\ \hline
$C_{34}$ & $\left( 0,0,6,3,0,3,1,3,3,4,1,7\right) $ & -3252 & 48\\ \hline
$C_{35}$ & $\left( 0,4,2,1,6,3,3,5,3,4,5,7\right) $ & -3192 & 48\\ \hline
$C_{36}$ & $\left( 0,0,2,3,4,3,1,3,3,4,5,7\right) $ & -3180 & 48\\ \hline
$C_{37}$ & $\left( 0,4,4,1,2,7,3,5,1,4,1,3\right) $ & -3156 & 48\\ \hline
$C_{38}$ & $\left( 2,0,2,3,2,5,3,3,7,4,3,1\right) $ & -3120 & 48\\ \hline
$C_{39}$ & $\left( 0,2,0,3,6,3,3,7,1,2,5,7\right) $ & -3036 & 48\\ \hline
$C_{40}$ & $\left( 0,2,0,1,4,3,1,1,1,6,5,7\right) $ & -3024 & 48\\ \hline
$C_{41}$ & $\left( 0,2,0,1,4,7,1,1,1,6,5,3\right) $ & -2976 & 48\\ \hline
$C_{42}$ & $\left( 0,0,4,3,0,7,1,3,1,4,1,3\right) $ & -2952 & 48\\ \hline
$C_{43}$ & $\left( 2,0,0,3,2,5,3,3,5,4,3,1\right) $ & -2868 & 48\\ \hline
\end{tabular}%
\end{table}

\begin{table}[H]
\caption{New Type II [72,36,12]
self-dual codes obtained from bordered double circulant matrices over $\R_{3,1}$}
\begin{tabular}{|c|c|c|c|c|}
\hline
Code $C_i$& first row of $M$ & $x,y,z$ & $\alpha$ in $W_{72}$ & $|Aut(C_i|$ \\ \hline
$C_{44}$ & $\left( 0,2,0,3,2,1,5,2,5,7,1\right) $ & $\left( 6,3,3\right) $ &
-4134 & 132\\ \hline
$C_{45}$ & $\left( 0,0,2,3,6,3,1,2,7,5,3\right) $ & $\left( 6,3,3\right) $ &
-4002 & 44\\ \hline
$C_{46}$ & $\left( 0,2,4,1,2,1,3,4,7,5,1\right) $ & $\left( 0,1,1\right) $ &
-3996 & 44\\ \hline
$C_{47}$ & $\left( 0,0,2,3,2,3,1,6,3,5,7\right) $ & $\left( 6,3,3\right) $ &
-3870 & 44\\ \hline
$C_{48}$ & $\left( 0,0,2,3,6,5,1,4,3,1,5\right) $ & $\left( 0,1,1\right) $ &
-3864 & 44\\ \hline
$C_{49}$ & $\left( 0,0,4,1,6,7,1,4,7,5,5\right) $ & $\left( 6,3,3\right) $ &
-3804 & 44\\ \hline
$C_{50}$ & $\left( 0,0,0,3,6,1,7,4,1,3,3\right) $ & $\left( 6,3,3\right) $ &
-3738 & 44\\ \hline
$C_{51}$ & $\left( 0,0,0,1,2,5,5,2,3,5,7\right) $ & $\left( 0,1,1\right) $ &
-3732 & 44\\ \hline
$C_{52}$ & $\left( 0,0,4,1,4,3,5,6,5,1,3\right) $ & $\left( 6,3,3\right) $ &
-3672 & 44\\ \hline
$C_{53}$ & $\left( 0,0,0,3,4,5,3,6,3,7,5\right) $ & $\left( 6,3,3\right) $ &
-3606 & 44\\ \hline
$C_{54}$ & $\left( 0,0,2,3,4,1,5,6,1,5,3\right) $ & $\left( 0,1,1\right) $ &
-3600 & 44\\ \hline
$C_{55}$ & $\left( 0,0,0,1,2,3,5,4,7,1,1\right) $ & $\left( 6,3,3\right) $ &
-3540 & 44\\ \hline
$C_{56}$ & $\left( 0,0,2,1,4,5,7,4,7,3,3\right) $ & $\left( 6,3,3\right) $ &
-3474 & 44\\ \hline
$C_{57}$ & $\left( 0,0,0,1,6,5,5,6,7,5,3\right) $ & $\left( 0,1,1\right) $ &
-3468 & 44\\ \hline
$C_{58}$ & $\left( 0,0,0,1,0,7,1,6,5,5,7\right) $ & $\left( 6,3,3\right) $ &
-3408 & 44\\ \hline
$C_{59}$ & $\left( 0,2,2,1,4,5,5,6,7,7,5\right) $ & $\left( 6,3,3\right) $ &
-3342 & 132\\ \hline
$C_{60}$ & $\left( 0,0,0,3,0,5,3,2,7,7,1\right) $ & $\left( 6,3,3\right) $ &
-3342 & 44\\ \hline
$C_{61}$ & $\left( 0,0,0,3,0,3,3,4,3,3,7\right) $ & $\left( 0,1,1\right) $ &
-3336 & 44\\ \hline
$C_{62}$ & $\left( 0,0,6,3,4,3,1,4,1,5,5\right) $ & $\left( 6,3,3\right) $ &
-3276 & 44\\ \hline
$C_{63}$ & $\left( 0,2,4,3,2,5,1,6,1,3,1\right) $ & $\left( 6,3,3\right) $ &
-3210 & 44\\ \hline
$C_{64}$ & $\left( 0,0,2,1,0,3,7,6,7,7,1\right) $ & $\left( 0,1,1\right) $ &
-3204 & 44\\ \hline
$C_{65}$ & $\left( 0,2,2,1,6,1,1,4,5,3,3\right) $ & $\left( 6,3,3\right) $ &
-3144 & 44\\ \hline
$C_{66}$ & $\left( 0,0,4,3,2,5,3,4,1,7,7\right) $ & $\left( 6,3,3\right) $ &
-3078 & 44\\ \hline
$C_{67}$ & $\left( 0,0,0,3,6,7,7,2,5,7,5\right) $ & $\left( 0,1,1\right) $ &
-3072 & 44\\ \hline
$C_{68}$ & $\left( 0,2,4,3,0,1,5,4,3,7,7\right) $ & $\left( 6,3,3\right) $ &
-2946 & 44\\ \hline
$C_{69}$ & $\left( 0,2,4,1,4,5,7,2,1,1,3\right) $ & $\left( 0,1,1\right) $ &
-2940 & 44\\ \hline
$C_{70}$ & $\left( 0,0,4,3,4,7,7,4,3,7,3\right) $ & $\left( 0,1,1\right) $ &
-2808 & 44\\ \hline
\end{tabular}%
\end{table}

\begin{table}[H]
\caption{New Type II [72,36,12]
self-dual codes obtained from bordered double circulant matrices over $\R_{3,2}$}
\begin{tabular}{|c|c|c|c|c|}
\hline
Code $D_i$& first row of $M$ & $x,y,z$ & $\alpha $ in $W_{72}$ & $|Aut(D_i)|$  \\ \hline
$D_{1}$ & $\left( 8,17,27,59,21\right) $ & $\left( 12,17,25\right) $ & -3960
& 120 \\ \hline
$D_{2}$ & $\left( 8,33,27,59,37\right) $ & $\left( 12,17,25\right) $ & -3960
& 40\\ \hline
$D_{3}$ & $\left( 8,1,11,43,5\right) $ & $\left( 12,17,25\right) $ & -3840
& 40\\ \hline
$D_{4}$ & $\left( 24,1,27,59,5\right) $ & $\left( 28,49,57\right) $ & -3732
& 40\\ \hline
$D_{5}$ & $\left( 24,17,27,59,21\right) $ & $\left( 28,33,41\right) $ & -3720
& 40\\ \hline
$D_{6}$ & $\left( 24,1,11,43,5\right) $ & $\left( 28,49,57\right) $ & -3612
& 40\\ \hline
$D_{7}$ & $\left( 8,1,27,59,5\right) $ & $\left( 12,17,25\right) $ & -3600
& 40\\ \hline
$D_{8}$ & $\left( 24,17,27,59,21\right) $ & $\left( 28,49,57\right) $ & -3492
& 40\\ \hline
$D_{9}$ & $\left( 8,33,11,43,37\right) $ & $\left( 12,17,25\right) $ & -3480
& 40\\ \hline
$D_{10}$ & $\left( 24,17,11,43,21\right) $ & $\left( 28,49,57\right) $ &
-3372 & 40\\ \hline
$D_{11}$ & $\left( 8,49,27,59,53\right) $ & $\left( 12,17,25\right) $ & -3360
& 40\\ \hline
$D_{12}$ & $\left( 56,17,57,43,21\right) $ & $\left( 28,49,49\right) $ &
-3252 & 40\\ \hline
$D_{13}$ & $\left( 8,17,11,43,21\right) $ & $\left( 12,17,25\right) $ & -3240
& 40\\ \hline
$D_{14}$ & $\left( 10,1,25,29,37\right) $ & $\left( 42,11,11\right) $ & -3120
& 40\\ \hline
$D_{15}$ & $\left( 10,1,25,29,37\right) $ & $\left( 42,35,35\right) $ & -3000
& 40\\ \hline
$D_{16}$ & $\left( 10,1,57,61,37\right) $ & $\left( 42,11,11\right) $ & -2880
& 40\\ \hline
$D_{17}$ & $\left( 56,33,27,59,37\right) $ & $\left( 28,33,49\right) $ &
-3942 & 20\\ \hline
$D_{18}$ & $\left( 24,49,11,43,53\right) $ & $\left( 28,33,57\right) $ &
-3882& 20 \\ \hline
$D_{19}$ & $\left( 24,1,11,43,5\right) $ & $\left( 28,33,57\right) $ & -3822
& 20\\ \hline
$D_{20}$ & $\left( 10,9,1,37,13\right) $ & $\left( 10,11,59\right) $ & -3786
& 20\\ \hline
$D_{21}$ & $\left( 24,17,27,59,21\right) $ & $\left( 28,33,57\right) $ &
-3762 & 20\\ \hline
$D_{22}$ & $\left( 10,17,41,45,53\right) $ & $\left( 10,11,59\right) $ &
-3726 & 20\\ \hline
$D_{23}$ & $\left( 56,1,11,43,5\right) $ & $\left( 28,33,49\right) $ & -3702
& 20\\ \hline
$D_{24}$ & $\left( 24,49,27,59,53\right) $ & $\left( 28,33,57\right) $ &
-3642 & 20\\ \hline
$D_{25}$ & $\left( 10,9,49,21,13\right) $ & $\left( 10,11,59\right) $ & -3606
& 20\\ \hline
$D_{26}$ & $\left( 56,33,11,43,37\right) $ & $\left( 28,33,49\right) $ &
-3582 & 20\\ \hline
$D_{27}$ & $\left( 10,1,9,13,37\right) $ & $\left( 10,11,59\right) $ & -3546
& 20\\ \hline
$D_{28}$ & $\left( 24,17,11,43,21\right) $ & $\left( 28,33,57\right) $ &
-3522 & 20\\ \hline
$D_{29}$ & $\left( 24,33,11,43,37\right) $ & $\left( 28,33,57\right) $ &
-3462 & 20\\ \hline
$D_{30}$ & $\left( 10,1,41,45,37\right) $ & $\left( 10,11,59\right) $ & -3426
& 20\\ \hline
$D_{31}$ & $\left( 56,49,27,59,53\right) $ & $\left( 28,33,49\right) $ &
-3402 & 20\\ \hline
$D_{32}$ & $\left( 10,9,17,53,13\right) $ & $\left( 10,11,59\right) $ & -3366
& 20\\ \hline
$D_{33}$ & $\left( 10,25,17,53,29\right) $ & $\left( 10,11,59\right) $ &
-3306  & 20\\ \hline
$D_{34}$ & $\left( 10,1,25,29,37\right) $ & $\left( 10,11,59\right) $ & -3246
& 20\\ \hline
$D_{35}$ & $\left( 10,9,33,5,13\right) $ & $\left( 10,11,59\right) $ & -3186
& 20\\ \hline
\end{tabular}%
\end{table}


\begin{thebibliography}{99}
\bibitem{betsumiya} K. Betsumiya, S. Georgiou, T.A. Gulliver, M. Harada and
C. Koukouvinos, ``On self-dual codes over some prime fields", Discrete Math,
vol. 262, pp. 37--58, 2003.

\bibitem{Magma} W. Bosma, J. Cannon and C. Playoust, The Magma algebra system. I. The user language, \emph{J. Symbolic Comput.}, vol. 24, pp. 235–-265, 1997.

\bibitem{bouyukliev} I. Bouyukliev, V. Fack and J. Winne, ``Hadamard matrices of
order $36$ and double-even self-dual $[72,36,12]$ codes", Eurocomb $2005$,
DMTCS proc., AE, pp. 93--98, 2005.

\bibitem{bouyuklieva} S. Bouyuklieva, ``Some optimal self-orthogonal and self-dual codes", Discrete Math., vol. 287, pp. 1--10, 2004.

\bibitem{conwaysloane} J. H. Conway and N. J. A. Sloane, ``A new upper bound
on the minimal distance of self-dual codes", IEEE Trans. Inform. Theory,
vol. $36$, no. $6$, pp. 1319--1333, 1990.

\bibitem{rdontcheva} R. Dontcheva, ``New binary self dual $[70,35,12]$ and
binary $[72,36,12]$ self dual doubly-even codes", Serdica Math. J., vol. 27, pp. 287--302, 2002.

\bibitem{doughertygulliverharada} S.T. Dougherty, A. Gulliver and M. Harada,
``Extremal binary self-dual codes", IEEE Trans. Infrom. Theory, vol. $43$, no.%
$6$, pp. 2036--2047, 1997.

\bibitem{doughertykim} S.T. Dougherty, J-L. Kim and P. Sole, ``Double circulant
codes from two class association schemes", Advances in Mathematics of
Communications, vol. 1, no.1, pp. 45--64, 2007.

\bibitem{doughertygaboritharada} S.T. Dougherty, P. Gaborit, M. Harada and P.
Sol\'{e}, ``Type \textrm{II} codes over $\mathbb{F}_{2}+u\mathbb{F}_{2}$",
IEEE Trans. Inform. Theory vol. 45, pp. 32--45, 1999.

\bibitem{yildizdoughertykaradeniz} S.T. Dougherty, B. Yildiz and S. Karadeniz,
``Codes over $R_{k},$Gray Maps and their Binary Images", Finite Fields Appl.,
vol. 17, pp. 205--219, 2011.

\bibitem{georgiou} S.D. Georgiou and E. Lappas, \textquotedblleft Self-dual
codes from circulant matrices\textquotedblright , Des. Codes Cryptogr., vol. $%
64$. pp. 129--141, 2012.

\bibitem{gulliverharada} T. A. Gulliver and M. Harada, ``On double circulant
doubly even self-dual $[72,36,12]$ codes and their neighbors", Australasian
Journal of Combinatorics, vol. 40, pp. 137--144, 2008.

\bibitem{doublecirculant} T. A. Gulliver and M. Harada, ``Classification of extremal double circulant self-dual codes of lengths 74--88", Discrete Math., vol. 306, pp. 2064--2072, 2006.

\bibitem{Golay} S. Karadeniz and B. Yildiz, ``A New Construction for the Extended Binary Golay Code", Appl. Math. Inf. Science, vol. 8, no. 1, pp. 69--72, 2014.

\bibitem{karadenizyildizaydin} S. Karadeniz, B. Yildiz and N. Aydin, ``Extremal
binary self-dual codes of lengths 64 and 66 from four-circulant
constructions over $\mathbb{F}_{2}+u\mathbb{F}_{2}$", {\bf to appear in} FILOMAT.

\bibitem{karadeniz66} S. Karadeniz, B. Yildiz, \textquotedblleft New
extremal binary self-dual codes of length $66$ as extensions of self-dual
codes over $R_{k}$",  J. Franklin Inst., vol. 350, no. 8, pp.
1963--1973, 2013.

\bibitem{karadeniz68} S. Karadeniz, B. Yildiz, \textquotedblleft New
extremal binary self-dual codes of length $68$ from $R_{2}$-lifts of binary
self-dual codes ", Advances in Mathematics of Communications, vol.
7, no. 2, pp. 219--229, 2013.

\bibitem{akaya} A. Kaya, B. Yildiz and I. Siap, ``New extremal binary self-dual
codes of length 68 from quadratic residue codes over $\mathbb{F}_{2}+u%
\mathbb{F}_{2}+u^{2}\mathbb{F}_{2}$", arXiv:$1308.0580$, 2013.

\bibitem{Hurley}
I.~McLoughlin and T.~Hurley, ``A Group ring construction of the
extended binary Golay code", IEEE Trans. Infrom. Theory,
vol. 54, pp. 4381--4383, 2008.

\bibitem{MacWilliams} F.J. MacWilliams, N. J. A. Sloane, ``The Theory of
Error-Correcting Codes", North-Holland Publishing Company,
1977.

\bibitem{melchorgaborit} C. A. Melchor, P. Gaborit, ``On the Classification
of Extremal $[36,18,8]$ Binary Self-Dual Codes", IEEE Trans. Inform. Theory,
vol. $54$, no. $10$, pp. 4743--4750, 2008.

\bibitem{Peng}
X.H.~Peng and P.~Farrell, ``On construction of the $(24,12,8)$ Golay
codes", IEEE Trans. Inform. Theory, vol. 52, pp. 3669--3675, 2006.

\bibitem{tufekci}
N. Tufekci, ``The generators of Type I self-dual codes of length 72", http://www.fatih.edu.tr/$\sim$ntufekci/type1codes.html.

\bibitem{jwood} J. Wood, ``Duality for modules over finite rings and
applications to coding theory", Amer. J. Math., vol. 121, pp. 555--575, 1999.

\bibitem{yildizkaradeniz} B. Yildiz, S. Karadeniz, ``Linear Codes over $%
\mathbb{F}_{2}+u\mathbb{F}_{2}+v\mathbb{F}_{2}+uv\mathbb{F}_{2}$", Des.
Codes Crypt. vol.54, pp. 61--81, 2010.

\bibitem{BYSK}
B.Yildiz and S.Karadeniz, ``Self-dual codes over
$\F_2+u\F_2+v\F_2+uv\F_2$", J. Franklin Inst.,
vol. 347, no. 10, pp.1888--1894, 2010.

\end{thebibliography}
\end{document}